\documentclass{llncs}
\usepackage{amsmath}
\usepackage{multirow}
\usepackage{tikz}
\usepackage{subfig}
\usepackage{url}
\usepackage{enumitem}
\usetikzlibrary{shapes.multipart}
\usepackage[a4paper]{geometry}
\usepackage{hyperref}

\def\hexdigit#1{\ifnum#1<10 \number#1\else
\ifnum#1=10 A\else\ifnum#1=11 B\else\ifnum#1=12 C\else
\ifnum#1=13 D\else\ifnum#1=14 E\else\ifnum#1=15 F\fi
\fi\fi\fi\fi\fi\fi}

\font\tenmsam=msam10

\font\sevenmsam=msam7
\font\fivemsam=msam5
\newfam\msafam
\textfont\msafam=\tenmsam
\scriptfont\msafam=\sevenmsam
\scriptscriptfont\msafam=\fivemsam

\font\tenmsbm=msbm10

\font\sevenmsbm=msbm7
\font\fivemsbm=msbm5
\newfam\msbfam
\textfont\msbfam=\tenmsbm
\scriptfont\msbfam=\sevenmsbm
\scriptscriptfont\msbfam=\fivemsbm
\def\mathbb{\fam=\msbfam\tenmsbm}

\mathchardef\N="0\hexdigit\msbfam4E
\mathchardef\R="0\hexdigit\msbfam52
\mathchardef\F="0\hexdigit\msbfam46
\mathchardef\Z="0\hexdigit\msbfam5A
\mathchardef\le="3\hexdigit\msafam36
\mathchardef\ge="3\hexdigit\msafam3E
\def   \Mod   #1{\;{\rm mod}\;#1}

\begin{document}

\title{Certificate Transparency with Enhancements \\ and Short Proofs
\thanks{This is the author's (full) version of the paper with the same title published in ACISP 2017 (available at \href{https://link.springer.com/chapter/10.1007/978-3-319-59870-3_22}{link.springer.com}).}}

\author{Abhishek Singh\inst{1}\thanks{Work done while at Indian Statistical Institute, Kolkata, India.}
\and
Binanda Sengupta\inst{2}
\and
Sushmita Ruj\inst{2}}

\index{Singh, Abhishek}
\index{Sengupta, Binanda}
\index{Ruj, Sushmita}
\institute{
IBM Research Laboratory, New Delhi, India\\
\email{absingh0@in.ibm.com}
\and
Indian Statistical Institute, Kolkata, India\\
\email{binanda\_r@isical.ac.in, sush@isical.ac.in}
}

\maketitle

\pagestyle{plain}

\begin{abstract}
Browsers can detect malicious websites that are provisioned with
forged or fake TLS/SSL certificates. However, they are not so good at
detecting malicious websites if they are provisioned with mistakenly
issued certificates or certificates that have been issued by a
compromised certificate authority.
Google proposed certificate transparency which is an open framework
to monitor and audit certificates in real time.
Thereafter, a few other certificate transparency schemes have been proposed which can even handle revocation.
All currently known constructions use Merkle hash trees and have proof size logarithmic in the number of certificates/domain owners.

We present a new certificate transparency scheme with short (constant size) proofs.
Our construction makes use of dynamic bilinear-map accumulators.
The scheme has many desirable properties
like efficient revocation, low verification cost and
update costs comparable to the existing schemes.
We provide proofs of security and evaluate the performance of our scheme.

\end{abstract}

\section{Introduction}
\label{sec:intro}
In public key encryption schemes, the sender encrypts a message using the receiver's public key
to produce a ciphertext, and the receiver decrypts the ciphertext to obtain the message using her private key.
On the other hand, in digital signature schemes, a verifier uses the public key of the signer
in order to check whether a signature on a given message is valid. Thus, in public key cryptography, a web user should be able to 
verify the authenticity of the public keys of different domains. Suppose, for example, a web user logs in
into her bank account through a web browser, and this web session is made secure by using the public key
of the concerned bank. If the web browser uses the public key of some attacker instead of
the bank's public key, then all the (possibly sensitive) information along with the login 
credentials may be known to the attacker who can misuse them later.

One solution to prevent such attacks is to rely on a third-party entity called {\em certificate authority} (CA)
that issues digital certificates showing the association of public keys with the domain owners.
The certificate authority signs each of these certificates using its private key.
When a web user wants to communicate with a server, she receives the server's certificate (signed by an appropriate CA). 
The process of verifying the ``signed certificate" is done by the web user's software 
(typically a web browser) that maintains an internal list of popular CAs and 
their public keys. It uses the appropriate public key to verify the signature on the server's certificate.
However, this CA model suffers from the following two major problems~\cite{MDRyan}. 
Firstly, if the CA is untrusted, then it may issue certificates which certify
the  ownership of fake public keys that could be created by an attacker or by the CA itself~\cite{DigiNotar}.
So, the CA must be trustworthy. Secondly, if the private key of a certificate owner
is compromised, then the CA must revoke the corresponding certificate before its expiry date.

In order to overcome these issues, researchers have come up with various solutions. 
Techniques like certificate pinning~\cite{langley2011public,marlinspike2013trust} and
crowd-sourcing~\cite{AlicherryK09,amann2012revisiting,eckersley2010ssliverse,WendlandtAP08}
have been proposed to restrict the browser to obtain certificates from the verified CAs only.
All the approaches discussed above are centralized where a certificate authority acts
as a trusted third-party responsible for managing digital certificates for
domain owners. An alternative approach to the problem of public authentication
of public key information is the ``Web of Trust''~\cite{WoTGermano}
(described by Phil Zimmerman in the context of PGP~\cite{WoTWiki}) that uses self-signed certificates
and third-party attestations of those certificates. It is entirely decentralized in that
a domain owner signs the public keys of other domain owners (whom it trusts) and designates them
as trustworthy.
However, it is difficult for a new domain owner joining the network
to get its public key signed by others.
Moreover, it does not deal with key revocation.

Certificate transparency (CT)~\cite{laurie2013rfc,CTWeb}, a technique
proposed by Google, aims to make certificate issuance transparent by
efficiently detecting fake certificates issued by malicious certificate authorities.
To achieve transparency, public append-only log structures are maintained containing all the
certificates. Domain owners can obtain proofs that their certificates are recorded 
in a log structure appropriately. Then, they provide the certificate along with a proof
to their clients so that the clients can be convinced about the authenticity of the received
certificate. Google's CT scheme provides two basic proofs:
\textit{proof of presence} (that is, the issued certificate is present in the log structure)
and \textit{proof of extension} (that is, the log structure is maintained in an append-only mode).
However, certificate transparency by
Google does not handle revocation of a certificate. Ryan~\cite{MDRyan}
extended Google's scheme to support certificate revocation efficiently and
provided two more proofs: \textit{proof of currency} (that is, the issued certificate
is current or active) and \textit{proof of absence of a domain owner} (that is, no certificates
have been issued for a particular domain owner). In both of the CT schemes described above,
the size of a proof, the computation time (and verification time) of the proof
are logarithmic in the number of certificates present in the log structure.
All these schemes are based on Merkle hash trees which have proofs in $O(\log n)$, where $n$ is the number of certificates. 
We take a different approach to have shorter proofs.
We achieve this efficiency by using bilinear-map accumulators~\cite{Damgard}.

On the other hand, certificate revocation lists (CRL) are required  in many applications.
Certificates can be revoked by CAs even before the expiry of the certificates (e.g., when the corresponding
private key is known to be compromised). In this work, we introduce a proof of absence of a suspected certificate
so that an auditor can ask for such proofs from the log maintainer for the certificates which belong to the CRL.
To the best of our knowledge, the proof of absence of a certificate has not been considered in any existing literature.

\begin{table}[t]
\footnotesize
\centering
\begin{tabular}{|c|c|c|c|c|c|c|}
\hline
		& \multirow{4}{*}{{\bf Parameters}} & {\bf Proof of} & {\bf Proof of} & {\bf Proof of} & {\bf Proof} & {\bf Proof}\\
{\bf Schemes} 	& & {\bf Presence of a} & {\bf Absence of a} & {\bf Absence of a} & {\bf of} & {\bf of}\\
{\bf for CT}	& & {\bf Certificate} & {\bf Certificate} & {\bf Domain Owner} & {\bf Extension} & {\bf Currency}\\
		& & {\bf (Type 1)} & {\bf (Type 2)} & {\bf (Type 3)} & {\bf (Type 4)} & {\bf (Type 5)}\\
\hline
\multirow{5}{*}{Google~\cite{laurie2013rfc}} & Proof Size & O(log $n$)& - & - & O(log $n$) & - \\ \cline{2-7}
&Cost of Proof & \multirow{2}{*}{O(log $n$)}& \multirow{2}{*}{-} & \multirow{2}{*}{-} & \multirow{2}{*}{O(log $n$)} & \multirow{2}{*}{-} \\ %\cline{3-7}
&Computation & &  &  &  &  \\ \cline{2-7}
&Cost of Proof & \multirow{2}{*}{O(log $n$)}& \multirow{2}{*}{-} & \multirow{2}{*}{-} & \multirow{2}{*}{O(log $n$)} & \multirow{2}{*}{-} \\ %\cline{3-7}
&Verification & &  &  &  &  \\ \cline{2-7}
\hline

\multirow{5}{*}{Ryan~\cite{MDRyan}} & Proof Size & - & - & O(log $t$)& O(log $n$) & O(log $t$) \\\cline{2-7}
&Cost of Proof & \multirow{2}{*}{-}& \multirow{2}{*}{-} & \multirow{2}{*}{O(log $t$)} & \multirow{2}{*}{O(log $n$)} & \multirow{2}{*}{O(log $t$)} \\ %\cline{3-7}
&Computation & &  &  &  &  \\ \cline{2-7}
&Cost of Proof & \multirow{2}{*}{-}& \multirow{2}{*}{-} & \multirow{2}{*}{O(log $t$)} & \multirow{2}{*}{O(log $n$)} & \multirow{2}{*}{O(log $t$)} \\ %\cline{3-7}
&Verification & &  &  &  &  \\ \cline{2-7}

\hline

& Proof Size & O(log $t)^\dag$ & O(1) & O(log $t$) & O(log $n$) & O(1) \\\cline{2-7}
Our &Cost of Proof & \multirow{2}{*}{O(log $t)^\dag$}& \multirow{2}{*}{O($m$)} & \multirow{2}{*}{O(log $t$)} & \multirow{2}{*}{O(log $n$)} & \multirow{2}{*}{O(log $m$)} \\ %\cline{3-7}
Scheme &Computation & &  &  &  &  \\ \cline{2-7}
&Cost of Proof & \multirow{2}{*}{O(log $t)^\dag$}& \multirow{2}{*}{O(1)} & \multirow{2}{*}{O(log $t$)} & \multirow{2}{*}{O(log $n$)} & \multirow{2}{*}{O(1)} \\ %\cline{3-7}
&Verification & &  &  &  &  \\ \cline{2-7}
\hline
\end{tabular}
\vspace{.05in}
\caption{Comparison among the existing certificate transparency schemes.
	 Here, $n$ is the total number of certificates present in the log structure, $t$ is the total number of domain owners
	 and $m$ is the total number of active certificates present in the log structure
	 ($n\ge t\ge m$).\smallskip\newline
	 $\dag$ For a proof of presence, a certificate is searched in the list of the most recent $N$ (constant) public keys
	 of the corresponding domain owner $u$. This list is maintained in a node (associated with $u$) of the searchTree.
	 Instead, if we want to search the certificate in the chronTree (which includes
	 all the historical certificates), then the value of this parameter
	 is exactly the same as that for Google~\cite{laurie2013rfc}.
	 }\label{tab:compar}
\end{table}

\bigskip
\noindent{\bf Our Contribution}\quad Our contributions are summarized as follows.

\begin{itemize}
  \item We have developed and extended the idea of enhanced certificate transparency proposed by Ryan~\cite{MDRyan}.
  We have designed a certificate transparency scheme (using bilinear-map accumulators and binary trees)
  that supports all the proofs found in the previous works.
  For the existing proofs, the parameters in our scheme are comparable to those
  proposed in the earlier schemes.\smallskip
  
  \item Some of our proofs are shorter (of constant size) than those proposed in previous works.
  
  \item In addition to the proof of currency, we have introduced another proof
  (\textit{proof of absence of a certificate}) related to certificate revocation.
  Both of these proofs are of constant size, and verification cost is also constant
  for them.\smallskip
  
  \item We have provided the security analysis and the performance evaluation of our scheme.
  We defer the detailed performance analysis of our construction in
  Section~\ref{sec:perf_ana}. For a quick review, we summarize
  the comparison of our construction with the existing schemes for certificate
  transparency in Table~\ref{tab:compar}.
\end{itemize}

The rest of the paper is organized as follows. In Section~\ref{sec:background},
we briefly discuss about the preliminaries and background related to our work.
Section~\ref{sec:our_scheme} describes the detailed construction of our scheme
along with the data structures used in the construction. 
In Section~\ref{sec:security_ana}, we provide the security analysis of our scheme.
Section~\ref{sec:perf_ana} provides the performance analysis of our scheme. 
In the concluding Section~\ref{sec:conclusion}, we summarize the work done
in this paper.

\section{Preliminaries and Background}
\label{sec:background}

\subsection{Notation}
\label{notation}
We take $\lambda$ to be the security parameter.
An algorithm $\mathcal{A}(1^\lambda)$ is a probabilistic polynomial-time
algorithm when its running time is polynomial in $\lambda$ and its output $y$
is a random variable which depends on the internal coin tosses of $\mathcal{A}$.
An element $a$ chosen uniformly at random from a set $S$ is denoted as $a\xleftarrow{R}S$.
A function $f:\N\rightarrow\R$ is called negligible in $\lambda$ if for
all positive integers $c$ and for all sufficiently large $\lambda$,
we have $f(\lambda)<\frac{1}{\lambda^c}$.

\subsection{Merkle Hash Tree}
\label{MHT}
A Merkle hash tree~\cite{Merkle_CR} is a binary tree where each leaf-node stores a data item.
The label of each leaf-node is the data item stored in the node itself.
A collision-resistant hash function $h$ is used to label the intermediate nodes of the tree.
The label of a intermediate node $v$ is the output of $h$ computed on the labels of the children
nodes of $v$.
A Merkle hash tree is used as a standard tool for efficient memory-checking.
Fig.~\ref{fig:merkletree} shows a Merkle hash tree containing the data items $\{d_1,d_2,\ldots,d_8\}$ 
stored at the leaf-nodes. Consequently, the labels of the intermediate nodes are computed
using the hash function $h$. The hash value of the root node $A$ (the root digest) is made public.
The proof showing that a data item $d$ is present in the tree consists of the data item $d$ and
the labels of the nodes along the \textit{associated path} (the sequence of siblings of the node
containing the data item $d$). For example, a proof showing that $d_3$ is present in the tree
consists of $\{d_3,(d_4,l_D,l_C)\}$, where $d_4,l_D$ and $l_C$ are the labels of the nodes
$K,D$ and $C$, respectively. Given such a proof, a verifier computes the hash value of the root.
The verifier outputs \texttt{accept} if the computed hash value matches with the public root digest;
it outputs \texttt{reject}, otherwise. The size of a proof is logarithmic in the number of data items
stored in the leaf-nodes of the tree.

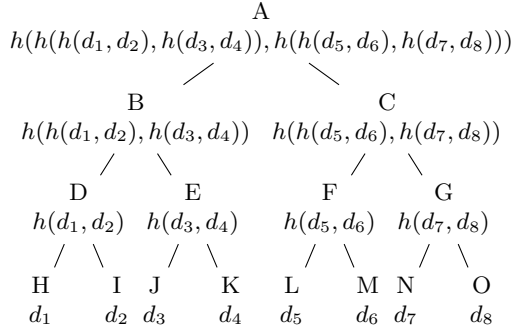
\begin{figure}[htbp]
\centering
\begin{tikzpicture}[level distance=1.2cm,
                      level 1/.style={sibling distance=3.3cm}, 
                      level 2/.style={sibling distance=1.5cm}, 
                      level 3/.style={sibling distance=1cm}, 
                      every text node part/.style={align=center}]
                    \node {A \\ $h(h(h(d_1, d_2), h(d_3, d_4)), h(h(d_5, d_6), h(d_7, d_8)))$}
                        child {node {B \\ $h(h(d_1, d_2), h(d_3, d_4))$}
                                child {node {D \\ $h(d_1, d_2)$}
                                    child {node {H \\ $d_1$}}
                                    child {node {I \\ $d_2$ }}
                                }
                                child {node {E \\ $h(d_3, d_4)$}
                                    child {node {J \\ $d_3$}}
                                    child {node {K \\ $d_4$}}
                                }
                        }
                        child {node {C \\ $h(h(d_5, d_6), h(d_7, d_8))$}
                                child {node {F \\ $h(d_5, d_6)$}
                                    child {node {L \\ $d_5$}}
                                    child {node {M \\ $d_6$}}
                                }
                                child {node {G \\ $h(d_7, d_8)$}
                                    child {node {N \\ $d_7$}}
                                    child {node {O \\ $d_8$}}
                                }
                        };
                    \end{tikzpicture}

\caption{A Merkle hash tree containing data items $\{d_1,d_2,\ldots ,d_8\}$.}
\label{fig:merkletree}
\end{figure}

Due to the collision-resistance property of $h$, it is impossible (except with some
probability negligible in the security parameter $\lambda$) to add a new data item
in the Merkle hash tree without changing the root digest of the tree.

\subsection{Bilinear Map}
\label{bmap}
Let $G_1,G_2$ and $G_T$ be multiplicative cyclic groups of prime order $p$.
Let $g_1$ and $g_2$ be generators of the groups $G_1$ and $G_2$, respectively.
A bilinear map is a function $e: G_1\times G_2\rightarrow G_T$ such that:
(1) for all $u\in G_1, v\in G_2, a\in\Z_p, b\in\Z_p$, we have $e(u^a,v^b)=e(u,v)^{ab}$
(bilinear property), and
(2) $e$ is non-degenerate, that is, $e(g_1,g_2)\not = 1$.
Furthermore, properties (1) and (2) imply that
$e(u_1\cdot u_2,v)=e(u_1,v)\cdot e(u_2,v)$
for all $u_1, u_2\in G_1, v\in G_2$.

If $G_1=G_2=G$, the bilinear map is symmetric; otherwise, asymmetric.
Unless otherwise mentioned, we consider only the bilinear maps which
are efficiently computable and symmetric.

\subsection{Bilinear-Map Accumulator}
\label{bmapacc}

A cryptographic accumulator is a one-way membership function. 
It answers a query to check whether an element is a member
of a set $X$ without revealing the individual members of $X$.
An accumulator scheme was introduced by Benaloh and de Mare~\cite{BenalohM} 
and further developed by Baric and Pfitzmann~\cite{Baric}.
Both of these constructions are based on RSA exponentiation
functions (secure under the \textit{strong RSA}
assumption) and provides a constant size membership witness
for any element in $X$ with respect to the accumulation value
denoted by $A(X)$~\cite{DerlerHS15}.
Universal accumulators~\cite{universalAccumulator} 
are designed to provide both membership and non-membership
witnesses of constant size. Camenisch and Lysyanskaya~\cite{dynamicAccumulator}
proposed a dynamic accumulator in which elements can
be efficiently added into or removed from the accumulator.
Whenever an element is inserted or deleted from $X$,
then this results in an update on $A(X)$ and the membership 
(and non-membership) witnesses.

Nguyen~\cite{Nguyen} constructed the first dynamic (but not universal) accumulator
based on bilinear maps. Later, Damg{\aa}rd and Triandopoulos~\cite{Damgard}
extended the work of Nguyen in order to provide both membership and non-membership
witnesses. This scheme is proved secure under the {\em q-strong Diffie-Hellman} assumption.
We use this accumulator in our construction.
We briefly describe this scheme as follows.

Let an algorithm BLSetup$(1^\lambda)$ output $(p,g,G,G_T,e)$ as the parameters of a
bilinear map, where $g$ is a generator of $G$.
Given a set $X = \{x_1,x_2,\ldots,x_n\}$, an accumulation function
$f_s(X) : 2^{\Z_p^*} \rightarrow G$ gives the accumulation value $A(X)$
defined as $f_s(X)=A(X)=g^{(x_1+s)(x_2+s)\ldots(x_n+s)}$,
where $s \xleftarrow{R} \Z_p^*$ is the secret trapdoor information.
The set $\{g^{s^i}|0 \le i \le q\}$ is public, where $q$ is an upper bound on $|X|$.

For any $x \in X$, the membership witness is defined as
$w_x = g^{\prod_{x_j \in X: x_j \neq x}(x_j+s)}$. The verifier can
verify the membership witness by checking the equation
\begin{equation*}
	e(w_x,g^x \cdot g^s) \stackrel{?}{=} e(A(X),g).
\end{equation*}

For any $y \notin X$, the non-membership witness is defined
as $\hat{w}_y=(w_y,v_y)$, where $v_y = - \prod_{x \in X} (x-y) \text{ mod }p$
and $w_y = g^{\frac{(\prod_{x \in X}(x+s))+v_y}{y+s}}$.
Now, the verifier can verify the non-membership witness
by checking the equation
\begin{equation*}
	e(w_y,g^y \cdot g^s) \stackrel{?}{=} e(A(X) \cdot g^{v_y},g).
\end{equation*}

Under the $q$-Strong Diffie-Hellman assumption, any probabilistic polynomial-time algorithm
$\mathcal{B}(1^\lambda)$, given any set $X$ ($|X| \le q$) and set $\{g^{s^i}|0 \le i \le q\}$,
finds a fake non-membership witness of a member of $X$ or a fake membership witness of a
non-member of $X$ with respect to $A(X)$ only with a probability negligible in $\lambda$
(measured over the random choice of $s \in \Z_p^*$ and the internal coin tosses of
$\mathcal{B}$) \cite{Damgard,Nguyen}.

\subsection{Certificate Transparency}
\label{CT}
Certificate transparency (CT)~\cite{laurie2013rfc,CTWeb} is a technique
proposed by Google in order to efficiently detect certificates
maliciously issued by certificate  authorities. A detailed description
of this open framework is given in~\cite{CTWeb}.
The framework consists of the following main components.

\begin{itemize}
    \item {\bf Certificate Log}: All the certificates that have been issued
    by certificate authorities are stored in append-only log structures.
    These log structures are maintained by log maintainers in an authenticated fashion using Merkle hash trees (see Section~\ref{MHT}).
    This enables a log maintainer to provide two types of verifiable cryptographic proofs: 
    (a) proof of presence (that is, the issued certificate is present in the log structure) and 
    (b) proof of extension (that is, the log structure is maintained in an append-only mode).

    \item {\bf Monitors}: Monitors are publicly run servers that look for suspicious certificates
    (illegitimate or unauthorized certificates, certificates with strange permissions or unusual certificate extensions)
    by contacting the log maintainers periodically. 
    Monitors also verify that all logged certificates are visible in the log structure.
    Certificate authorities or domain owners can check the validity of a certificate with the help of monitors.

    \item {\bf Auditors}: Auditors are lightweight software components that
    can verify that logs are behaving correctly and are cryptographically consistent. 
    They can also check whether a particular certificate is recorded in a log appropriately.
    An auditor may be a part of a browser’s TLS client or a standalone service or a secondary function of a monitor.
\end{itemize}

We consider a certificate $c=cert(u, pk_u)$ as a signed pair $(u, pk_u)$, where $u$ is a domain owner and $pk_u$ is a 
public key of $u$.
A log maintainer maintains certificates issued by CAs in a Merkle hash tree that
stores the certificates as the leaf-nodes in chronological order (left-to-right).
A certificate authority or a domain owner can request the log maintainer to include a particular certificate.
Once the certificate maintainer logs in the certificate, it sends a ``Signed Certificate Timestamp'' (SCT)
to the CA or the domain owner via one of the following methods: X.509v3 extension,
TLS extension or OCSP (Online Certificate Status Protocol) stapling~\cite{CTWebMethod}.
On the other hand, gossip protocols \cite{NGR15} allow web users (browsers)
and domain owners to share information that they receive from log maintainers. 
Chuat \emph{et al.}~\cite{ChuatSPLM15} presented an efficient gossip protocol to verify consistency of log maintainers.
Some other works related to certificate transparency are discussed in Section~\ref{sec:intro} and Section~\ref{ECT}.

Whenever a new certificate is added, it is appended to the right of the tree.
The browser of a web user accepts the certificate only if it is accompanied by a proof of
the presence of the in the log. For proof of extension, the monitor
submits two hash values (computed at different time) of the log to the CA.
The CA returns a proof that one of them is an extension of the other.
Both the proof of presence and the proof of extension can be done in time/space
complexity of $O(\log n)$, where $n$ is the total number of certificates.

\section{Related work}
\label{ECT}
Enhanced certificate transparency (ECT) by Ryan~\cite{MDRyan} proposes an idea to address
the revocation problem that was left open by Google. It provides transparent key revocation.
It also reduces reliance on trusted parties by designing the monitoring role so that it 
can be distributed among the browsers.
In this extension, Ryan introduced two proofs: (a) proof of currency (that is, the certificate
is issued and not revoked) and (b) proof of absence of a domain owner (that is, the CA has not issued
any certificate for a particular domain owner). 
Both of these proofs are logarithmic in the number of certificates.
The proof of extension remains the same as that of Google's certificate transparency.

Recently, Dowling \emph{et al.}~\cite{DowlingGHS16} proposed a formal framework of secure logging in general.
Formal security proofs for secure logging have been presented. The underlying framework consider Google's framework
and does not support revocation of domain owners. 
However, only the proof of presence has been considered in this work. 
A different line of work is accountable PKI (APKI)~\cite{KimHPJG13,BasinCKPSS14}, which has been used
to address the decrease in trust on CAs. APKI efficiently handles revocation and is still a tree-based protocol.
Deployment challenges for log based PKI enhancements have been presented in \cite{MatsumotoSP15}.
Accountability in future internet has been discussed in \cite{BechtoldP14}. 
Key transparency is another closely related line of work.
It differs from the certificate transparency in that, in key transparency, there is a requirement
to protect the secrecy of the public keys of the domain owners.
Coniks~\cite{MelaraBBFF15} is an example of a key transparency scheme.

All these schemes maintain central log maintainers. Some Blockchain based solutions are also available.
A decentralized PKI setting has been proposed in~\cite{fromknecht2014certcoin}
where domain owners are incentivized in a cryptocurrency called Certcoin to perform audits.
The blockchain implementation of Coniks on Ethereum
appeared in \cite{B16}. Wilson and Ateniese recently exploited Bitcoin and its blockchain
to store and retrieve Bitcoin-based PGP certificates~\cite{Ateniese_PGP}.

\section{Our Construction}
\label{sec:our_scheme}
In our extension of certificate transparency, each certificate issued by
a (possibly malicious) certificate authority (CA) is associated with proofs showing the validity
of that particular certificate. The certificates issued by various CAs are stored in a public
(and append-only) log structure that is maintained by the log maintainer. This log maintainer
that maintains the public log of certificates issued by certificate authorities is known
as a \textit{certificate prover} (CP). We use the terms ``log maintainer'' and ``certificate prover''
interchangeably in this work. Anyone can request the log maintainer to get certificates logged in
the log structure.
Similarly, anyone can query the log structure for a cryptographic proof
to verify that the log maintainer is behaving properly or to verify that a particular certificate 
has been logged.

The number of log maintainers is typically small (e.g., much less than one thousand worldwide).
These servers are operated independently by certificate provers (CPs) that could be browser companies,
Internet service providers (ISPs) or any other interested parties.
In general, a certificate prover is a server which has enough resources 
(in terms of storage capacity and computing power) to store all the certificates
and compute the proofs relevant to certificate transparency efficiently.

\subsection{Data Structures Used in Our Construction}
\label{subsec:our_idea}
In our construction, the public log structure maintained by the certificate prover is organized by
using the following tree data structures: \textit{chronTree}, \textit{searchTree} and \textit{accTree}.
The first two of them are similar to the data structures described in~\cite{MDRyan}.
Let $n$ be the total number of certificates present in the log structure,
$t$ be the total number of domain owners and $m$ be the total number of active certificates
present in the log structure. Clearly, $n \ge t \ge m$.
We describe these data structures as follows.
They are illustrated in Fig.~\ref{fig:example}.

\begin{figure*}%[htbp]
\centering
	    \subfloat[The chronTree stores the certificates $c_i$ in the chronological order
	    (the higher the value of $i$, the more recent the certificate $c_i$ is).
	    The certificates $c_1=cert(Alice, pk_{Alice})$, $c_2=cert(Bob, pk_{Bob})$, $c_3=cert(Alice, \texttt{null})$,
	    $c_4=cert(Charlie, pk_{Charlie})$, $c_5=cert(Alice, pk_{Alice}')$, $c_6=cert(Bob, \texttt{null})$,
	    $c_7=cert(Eve, pk_{Eve})$, $c_8=cert(Bob, pk_{Bob}')$, $c_9=cert(Frank, pk_{Frank})$, and
	    $c_{10}=cert(Henry, pk_{Henry})$ are stored in the order they are issued (or revoked).
	    When the certificate $c_1$ (or $c_2$) is revoked, another certificate $c_3$ (or $c_6$) is inserted in the chronTree
	    having \texttt{null} as the public key of $Alice$ (or $Bob$).
	    The leaf-nodes of the chronTree contain items of the form $x = (c, A, digST)$,
	    where $A$ is the accumulation value of the updated set $X$ after
	    inserting the certificate $c$ or revoking the active certificate of the corresponding domain owner.]
	    {
	    \centering
		  \begin{tikzpicture}[level distance=1cm,
                      level 1/.style={sibling distance=4.9cm}, 
                      level 3/.style={sibling distance=1.5cm}, 
                      level 4/.style={sibling distance=.7cm},
		      every text node part/.style={align=center}]
                   \node {$digCT = h(h(h(h(x_1, x_2), h(x_3, x_4)), h(h(x_5, x_6), h(x_7, x_8))),h(x_9,x_{10}))$}
                        child {node {$h(h(h(x_1, x_2), h(x_3, x_4)), h(h(x_5, x_6), h(x_7, x_8)))$}
                                child [sibling distance=3.7cm]{node {$h(h(x_1, x_2), h(x_3, x_4))$}
                                    child {node {$h(x_1, x_2)$}
					  child {node {$x_1$}}
					  child {node {$x_2$}}
                                    }
                                    child {node {$h(x_3, x_4)$}
					  child {node {$x_3$}}
					  child {node {$x_4$}}
                                    }
                                }
                                child [sibling distance=3cm]{node {$h(h(x_5, x_6), h(x_7, x_8))$}
                                    child {node {$h(x_5, x_6)$}
					  child {node {$x_5$}}
					  child {node {$x_6$}}
                                    }
                                    child {node {$h(x_7, x_8)$}
					  child {node {$x_7$}}
					  child {node {$x_8$}}
                                    }                                
                               }
                        }
                        child {node {$h(x_9, x_{10})$}
                              child [sibling distance=.7cm]{node {$x_9$}}
                              child [sibling distance=.7cm]{node {$x_{10}$}}
                        };
                        
                \end{tikzpicture}
                }\bigskip\bigskip

                \subfloat[The searchTree stores the certificates in the lexicographic order of the domain owners.]{
		\centering
                    \begin{tikzpicture}[level distance=1.2cm,
                      level 1/.style={sibling distance=4.1cm},
                      level 2/.style={sibling distance=3.3cm}, every text node part/.style={align=center}]
                    \node {$d_4 = (Eve, pk_{Eve})$ \\ $digST_8 = h(d_4,h(d_2, h(d_1), h(h_3)), h(d_5, h(d_6)))$}
                        child {node {$d_2 = (Bob, (pk_{Bob}, pk^\prime_{Bob}))$ \\ $h(d_2, h(d_1), h(h_3))$}
                                    child [level distance = 1.2cm]{node {$d_1 = (Alice, (pk_{Alice},pk^\prime_{Alice}))$ \\ $h(d_1)$}}
                                    child [level distance = 2cm]{node {$d_3 = (Charlie, pk_{Charlie})$ \\ $h(d_3)$}}
                        }
                        child {node {$d_5 = (Frank, pk_{Frank})$ \\ $h(d_5, h(d_6))$}
                                    child[missing]{ node {}}
                                    child [level distance = 1.2cm]{node {$d_6 = (Henry, pk_{Henry})$ \\ $h(d_6)$}}
                        };
                    \end{tikzpicture}
                }\bigskip\bigskip

                \subfloat[The accTree stores the elements of $X$ (the set of active or current certificates)
	and their corresponding membership (in $X$) witnesses.]{
		\centering
                    \begin{tikzpicture}[level distance=.9cm,
                      level 1/.style={sibling distance=4.4cm},
                      level 2/.style={sibling distance=3.2cm}, every text node part/.style={align=center}]
                    \node {$c_7=cert(Eve, pk_{Eve}),w_{c_7}$}
                        child {node {$c_8=cert(Bob, pk'_{Bob}),w_{c_8}$}
                                    child [level distance = 0.9cm]{node {$c_5=cert(Alice, pk'_{Alice}),w_{c_5}$}}
                                    child [level distance = 1.4cm]{node {$c_4=cert(Charlie, pk_{Charlie}),w_{c_4}$}}
                        }
                        child {node {$c_9=cert(Frank, pk_{Frank}),w_{c_9}$}
                                    child[missing]{ node {}}
                                    child [level distance = 0.9cm]{node {$c_{10}=cert(Henry, pk_{Henry}),w_{c_{10}}$}}
                        };
                    \end{tikzpicture}
                }\bigskip\bigskip
\caption{The structures of chronTree, searchTree and accTree used in our scheme.}
\label{fig:example}
\end{figure*}

\begin{itemize}
    \item {\bf chronTree}: The chronTree is a Merkle hash tree %(see Appendix~\ref{MHT})
    where certificates are stored as the leaf-nodes of the tree. The certificates are
    arranged in the chronological order in left-to-right manner. When a new certificate
    $c=cert(u, pk_u)$ is issued by a CA, it is added to the right of the chronTree. When a certificate
    $c=cert(u, pk_u)$ is revoked, another certificate $c'=cert(u, \texttt{null})$
    is added to the right of the chronTree. A collision-resistant hash
    function $h$ is used to compute the hash values of the nodes of the 
    chronTree. The hash value of the root node (the root digest) of the chronTree
    is denoted by $digCT$.

    \item {\bf searchTree}: This tree is organized as a modified binary search tree where
    data items corresponding to the domain owners are stored in the lexicographic order (of the domain owners).
    Here, a data item corresponding to a domain owner $u$ is of the form $(u, List(pk_{u}))$,
    where $List(pk_u)$ is the list of $N$ most recent public keys of the domain owner $u$.
    In other words, the last certificate in the list is the current public key
    of the domain owner, and other keys are already revoked. The value of $N$ is taken to be
    constant, and the list is maintained in a first-in-first-out (FIFO) fashion. The data items
    are stored in leaf-nodes as well as in non-leaf nodes such that an in-order traversal
    of the searchTree provides the lexicographic ordering of the domain owners. The collision-resistant
    function $h$ is used to compute the hash values corresponding to the nodes of the 
    searchTree. The hash value of a node is computed on the data item (of that node)
    and the hash values of its children. This hash value is also stored in the node
    along with the data item.
    The hash value of the root node (the root digest) of the searchTree is denoted by $digST$.
    The value of $digST$ is linked to a leaf-node of the chronTree.

    \item {\bf accTree}: This tree is organized as a binary search tree in which active
    certificates are stored in the lexicographic order of the domain owners.
    Let $X$ be the set of active (or current) certificates for different domain owners.
    In our construction, the set $X$ is implemented as accTree.
    Each node in the accTree contains a certificate $c=cert(u, pk_{u})\in X$ and the corresponding
    membership (in $X$) witness $w_c$ of $c$.
    The accumulation value $A(X)$ is linked to a leaf-node of the chronTree.
\end{itemize}

\subsection{Detailed Construction}
\label{subsec:contribution}
In this section, we describe our construction in details.
Our construction involves the following algorithms to achieve certificate transparency.

\begin{itemize}
    \item \textbf{Setup}($1^\lambda$): The Setup algorithm runs BLSetup$(1^\lambda)$ to output
    $(p,g,G,G_T,e)$ as the parameters of a bilinear map, where $g$ is a generator of $G$.
    Let $X$ be the set of active certificates issued by a certificate authority (CA),
    that is,
    \begin{equation*}
        X = \{cert(u_i, pk_{u_i})\},
    \end{equation*}
    where $pk_{u_i}$ is the active public key issued by the CA for the domain owner $u_i$.
    The Setup algorithm selects a random element $s\xleftarrow{R}{\Z}_p^*$ as the secret trapdoor information.
    The set $\{g^{s^i}|0 \le i \le q\}$ is made public, where $q$ is an upper bound on $|X|$.
    The accumulation function $f_s(X) : 2^{{\Z}_p^*} \rightarrow G$ gives the accumulation value
    $A(X)$ defined as
    \begin{equation*}
        f_s(X) = A = g^{\prod_{x_i \in X}(x_i + s)}.
    \end{equation*}
    
    The algorithm constructs a searchTree by inserting domain owners in the lexicographic order
    along with other relevant data associated with each domain owner
    and returns $digST$ as the root digest of the searchTree.
    The algorithm constructs an accTree by inserting (only) the active certificates represented as set {\em X} 
    along with their membership (in $X$) witnesses for different domain owners.
    For each active certificate $c \in X$, the membership witness $w_c$ is computed as
    \begin{equation*}
        w_c = g^{\prod_{x_j \in X: x_j \neq c}(x_j+s)}= A^\frac{1}{(c+s)}.
    \end{equation*}
    Finally, the algorithm constructs a chronTree by inserting certificates in the chronological
    (left-to-right) order and returns $digCT$ as the root digest of the chronTree.
    These data structures are discussed in Section~\ref{subsec:our_idea}.
    We note that a \textit{collision-resistant} hash function $h$ is used to compute
    the hash values in the searchTree and the chronTree. Finally, $(p,g,G,G_T,e,\{g^{s^i}|0 \le i \le q\},h,A,digCT,$ $digST)$
    is set as the public parameters $PP$,
    and the secret key is the trapdoor value $s$.\bigskip
    
    \item \textbf{Insert}($c, sk, PP$): When a new certificate $c=cert(u, pk_{u})$ is issued by a CA,
    then it asks the log maintainer (or the certificate prover) to insert the certificate in the log structure.
    The new certificate $c$ is added to the log structure (accumulator, searchTree and chronTree)
    as follows. The public parameters $PP$ are updated accordingly.
    
    \begin{itemize}
        \item Adding $c$ to accTree: Compute the new accumulation value $A'$
        (corresponding to the new set $X'$ = $X\cup \{c\}$) as
        \begin{equation*}
            A' = A^{(c+s)}.
        \end{equation*}
        The membership witness for $c$ is $A$. For each $i \in X$, the updated membership
        witness is computed as
        \begin{equation*}
            w_i^\prime = w_i^{(c+s)}.
        \end{equation*}
        The accTree is updated accordingly.
        
        \item Adding $c$ to searchTree: Search for the node corresponding to the domain owner $u$, if it is present, 
        then append the new public key $pk_{u}$ to the associated list of public keys for $u$.
        Otherwise, create a new node for $u$ with the list containing only the public
        key $pk_{u}$ and insert it in the searchTree maintaining the lexicographic order. Consequently, 
        the root digest of the searchTree is updated as $digST'$.
        
        \item Adding $c$ to chronTree: Add a new node containing $(c, A', digST')$
        to the right of the existing chronTree. The new root digest of the chronTree is updated as $digCT'$.%\medskip
    \end{itemize}

    \item \textbf{Revoke}($c, sk, PP$): Let $c=cert(u, pk_{u})$ be the certificate to be revoked.
    Then, the following operations are performed on the log structure, and
    the public parameters $PP$ are updated accordingly. 
    
    \begin{itemize}
        \item Removing $c$ from accTree: Compute the new accumulation value $A'$
        (corresponding to the new set $X'$ = $X \backslash \{c\}$) as
        \begin{equation*}
            A' = A^{\frac{1}{(c+s)}}.
        \end{equation*}
        Remove the node corresponding to the domain owner $u$ of certificate $c$ from the accTree. 
        For each $i \in X'$, the updated membership witness is computed as
        \begin{equation*}
            w_i^\prime = w_i^{\frac{1}{(c+s)}}.
        \end{equation*}
        The accTree is updated accordingly.
        
        \item There are no changes in the searchTree for the revocation of $c$.
        Therefore, the value of $digST$, the root digest of the searchTree, remains the same.
        
        \item Adding a new node for the domain owner $u$ to chronTree: Add
        to the right of the existing chronTree a new node containing $(c',A',digST)$, where $c'=cert(u,\texttt{null})$.
        The new root digest of the chronTree is updated as $digCT'$.%\medskip
    \end{itemize}

    \item \textbf{Query}($PP$): This algorithm is run by an auditor to output a query $Q$.
    The type of the query $Q$ is based on the type of the corresponding proof.
    
    \begin{itemize}
        \item Proof of presence of a certificate (Type 1): The query $Q$ asks for a proof of whether a certificate
	      $c = cert(u, pk_u)$ is present in the log structure.
    
        \item Proof of absence of a certificate (Type 2): The query $Q$ asks for a proof of whether a certificate
	      $c = cert(u, pk_u)$ is absent in the set of active certificates, that is, $c \notin X$. 
        
        \item Proof of absence of a domain owner (Type 3): The query $Q$ asks for a proof of whether
	      a domain owner $u$ is absent, that is, there are no certificates for $u$ in the log structure. 
    
        \item Proof of extension (Type 4): The query $Q$ asks for a proof of whether
	      the chronTree corresponding to $digCT'$ is an extension of the chronTree
	      corresponding to $digCT$.
    
        \item Proof of currency (Type 5): The query $Q$ asks for a proof of whether
	      $pk_u$ is the current public key of the domain owner $u$, that is, whether the certificate
	      $c=cert(u, pk_{u})$ is present in the set of active certificates ($c \in X$).%\medskip
	      
    \end{itemize}

    \item \textbf{ProofGeneration}($Q, PP$): Upon receiving the query $Q$, the certificate prover (CP)
    generates the corresponding proof $\Pi(Q)$ as follows.
    
    \begin{itemize}
        \item Proof of presence of a certificate (Type 1): 
        Search for the certificate $c$ in the searchTree. If a node for the domain owner $u$ is present
        in the searchTree, define $h_1$ and $h_2$ to be the hash values of the children of the node
        (they are taken to be \texttt{null} if the node is a leaf-node).
        
        Let the sequence of data items of the nodes along the search path be
        \begin{equation*}
            dataseq_{type1} = (d_1, d_2, d_3,\ldots,d_r)
        \end{equation*}
        for some $r\in\N$, where $d_1$ is the data item corresponding to the node for the domain owner $u$,
        $d_r$ is the data item corresponding to the root node, and other data items
        correspond to the other intermediate nodes in the search path.
        Let the sequence of hash values of the nodes
        in the \textit{associated path} (the path containing the siblings of the nodes
        along the search path mentioned above) along with $h_1$ and $h_2$ be
        \begin{equation*}
            hashseq_{type1} = (h_1, h_2, h(v_1), h(v_2), h(v_3),\ldots).
        \end{equation*}
        Send these sequences as the corresponding proof $\Pi$.

        \item Proof of absence of a certificate (Type 2): Search for the certificate $c$ in the accTree.
        If there is no node for the domain owner $u$ in the accTree, then send the non-membership witness
        $\hat{w}_c=(w_c, v_c)$ of $c$, where
        \begin{equation*}
        	v_c = - \prod_{x \in X} (x-c) \Mod p \in\Z_p^*
        \end{equation*}
        and 
        \begin{equation*}	
        	w_c = g^{\frac{(\prod_{x \in X}(x+s))+v_c}{c+s}}\in G.
        \end{equation*}
        
        \item Proof of absence of a domain owner $u$ (Type 3): The proof is similar to the proof of Type 1.
        Find the nodes in the searchTree corresponding to the domain owners $u_1$ and $u_2$
        such that they were the neighbor (in the lexicographic ordering) nodes of the node
        corresponding to $u$ if $u$ were present in the searchTree, that is,
        $u_1 \le u \le u_2$ lexicographically. These nodes can be found by searching
        for the domain owner $u$ in the searchTree, and the search ends at some leaf-node in the searchTree.
        The nodes corresponding to $u_1$ and $u_2$ reside on this search path itself, and
        one of them is the leaf-node (where the search ends).
        
        Let the sequence of data items of the nodes along the search path be
        \begin{equation*}
            dataseq_{type3} = (d_1, d_2,\ldots,d_{r'})
        \end{equation*}
        for some $r'\in\N$, where $d_1$ is the data item corresponding to the leaf-node,
        $d_{r'}$ is the data item corresponding to the root node, and other data items
        correspond to the other intermediate nodes in the search path.
        Let the sequence of hash values of the nodes
        in the \textit{associated path} (the path containing the sibling nodes of the nodes
        along the search path) be
        \begin{equation*}
            hashseq_{type3} = (h(v_1), h(v_2),\ldots).
        \end{equation*}
        Send $(dataseq_{type3},hashseq_{type3})$ as the proof $\Pi$.

        \item Proof of extension (Type 4):
        Compare the chronTree structures corresponding to both $digCT$
        and $digCT'$ and send one hash value per level of the latest chronTree as a proof $\Pi$.
        If the chronTree corresponding to $digCT'$ is an extension of the chronTree corresponding
        to $digCT$, then the latter chronTree is a subtree of the earlier chronTree.
        The proof $\Pi=(h_1, h_2, \ldots)$ is the sequence of hash values of the nodes
        required to compute the current root digest $digCT'$ from the previous root digest $digCT$.
        Here, $h_1$ is the hash value of the sibling node of node $v$ whose
        hash value is $digCT$ (that is, $v$ is the root of the previous chronTree),
        $h_2$ is the hash value of the sibling node of parent node of $v$, and so on.
        
        \item Proof of currency (Type 5):
        Search for the certificate $c$ in the accTree. If $c$ is present
        in the node for the domain owner $u$ in the accTree, 
        then send the membership (in $X$) witness $w_c$ stored at that node.%\medskip
    \end{itemize}
    
    \item \textbf{Verify}($Q, \Pi, PP$): Given the query $Q$ and the corresponding proof $\Pi$,
    the auditor verifies the proof in the following way depending on the type of the proof.
    
    \begin{itemize}
        \item Proof of presence of a certificate (Type 1):
        The proof consists of the sequences of data items
        %\begin{center}
        $dataseq_{type1} = (d_1, d_2,d_3,\ldots,d_r)$ 
        %\end{center}
        and
        the sequence of hash values
        $hashseq_{type1} = (h_1,h_2,h(v_1), h(v_2), h(v_3),\ldots)$.
        Given these two sequences, the auditor verifies whether
        \begin{equation*}
	  h(\cdots h(d_3,h(d_2,h(d_1, h_1, h_2),h(v_1)),h(v_2))\ldots)\stackrel{?}= digST
        \end{equation*}
        and outputs \texttt{accept} if the equation holds; it outputs \texttt{reject}, otherwise.
        
        \item Proof of absence of a certificate (Type 2):
        Given the non-membership (in $X$) witness $\hat{w}_c=(w_c, v_c)$ for a certificate $c$,
        the value of $g^s$ (included in $PP$) and the accumulation value $A=f_s(X)$, the auditor verifies whether
        \begin{equation*}
             e(w_c,g^c \cdot g^s) \stackrel{?}= e(A\cdot g^{v_c},g)
        \end{equation*}
        and outputs \texttt{accept} if the equation holds; it outputs \texttt{reject}, otherwise.

        \item Proof of absence of a domain owner (Type 3):
        %The proof consists of the sequences of data items
        Given the sequence of data items
        $dataseq_{type3} = (d_1, d_2,\ldots,d_{r'})$ and
        the sequence of hash values
        $hashseq_{type3} = (h(v_1), h(v_2),\ldots)$. The auditor verifies whether
        \begin{equation*}
            h(\cdots h(d_3,h(d_2,h(d_1),h(v_1)),h(v_2))\ldots)\stackrel{?}= digST
        \end{equation*}
        and outputs \texttt{accept} if the equation holds; it outputs \texttt{reject}, otherwise.

        \item Proof of extension (Type 4):
        Given the proof $\Pi = (h_1, h_2, \ldots)$, the auditor verifies whether
        \begin{equation*}
            h(\ldots(h(h(digCT,h_1),h_2)\ldots) \stackrel{?}= digCT'
        \end{equation*}
        and outputs \texttt{accept} if the equation holds; it outputs \texttt{reject}, otherwise.

        \item Proof of currency (Type 5):
	Given the membership (in $X$) witness $w_c$ for a certificate $c$, the value of
        $g^s$ (included in $PP$) and the accumulation value $A=f_s(X)$, the auditor verifies whether
        \begin{equation*}
            e(w_c,g^c \cdot g^s) \stackrel{?}= e(A,g)
        \end{equation*}
        and outputs \texttt{accept} if the equation holds; it outputs \texttt{reject}, otherwise.
        
    \end{itemize}
    
\end{itemize}

\paragraph{\bf Note}\quad
Once the log maintainer receives a request from a domain owner or a certificate authority for the insertion
(or revocation) of a certificate,
it sends a ``Signed Certificate Timestamp'' (SCT) to the domain owner or the CA that made the request.
This SCT contains the timestamp at which the request was received.
It is a ``promise" that the certificate will be logged in future. 
The log maintainer handles these requests periodically and performs batch updates in general.
We note that web users and domain owners share information that they receive from log maintainer
by executing a gossip protocol~\cite{NGR15,ChuatSPLM15}.

\section{Security Analysis}
\label{sec:security_ana}
We assume that an auditor, a domain owner and a monitor are \textit{honest} while a log maintainer (or certificate prover) and a certificate authority
may be \textit{dishonest}. We further assume that the log maintainer does \textit{not} collude with the certificate authority.
We analyze the security of our scheme in the following lemmas.

\begin{lemma}
Let a certificate authority issue a fake certificate $c = cert(u, pk^{\prime}_u)$ for a
particular domain owned by $u$. If the certificate is not logged in the public log maintained by an honest log maintainer,
then an auditor will reject the certificate.
\end{lemma}

\begin{proof}
The above lemma holds because the auditor accepts the certificate $c = cert(u, pk^{\prime}_u)$ only when it is accompanied by a proof of 
presence of the certificate in the log. As the certificate $c$ is not present in the log, the log maintainer fails to provide
a proof of presence of the certificate.
\end{proof}

\begin{lemma}
Let a certificate authority issue a fake certificate $c = cert(u, pk^{\prime}_u)$ for a
particular domain owned by $u$. If the certificate is present in the public log maintained by an honest log maintainer,
then the domain owner will be able to immediately identify this certificate issued maliciously and to report this problem.
\end{lemma}

\begin{proof}
As the certificate $c = cert(u, pk^{\prime}_u)$ is present in the log, the domain owner $u$
(while running a \textit{monitor} on the log structure) can easily check that the new certificate
(that was not requested by $u$) is issued for that particular domain. %owner $u$ belonging to its domain.
Therefore, the domain owner reports this problem and asks the corresponding certificate authority
to revoke this particular $c$. 
\end{proof}

\begin{lemma}
Let a log maintainer (or certificate prover) be honest. Let a dishonest certificate authority
issue a fake certificate $c = cert(u, pk^{\prime}_u)$ for a particular domain owned by $u$.
If the certificate $c$ is not present in the log, then the certificate authority fails to produce
a valid proof $\Pi$ of any type, except with some probability negligible in the security parameter $\lambda$.
\end{lemma}

\begin{proof}
For this scenario, we define the security model and prove our scheme to be secure in this model.

\paragraph{\bf Security Model}\quad
We define the security game between the challenger (acting as the certificate prover) 
and the probabilistic polynomial-time adversary (acting as the certificate authority) as follows.

\begin{itemize}
\item The challenger executes the Setup algorithm to generate the secret information $s$ 
for the bilinear-map accumulator and the public parameters $PP$.
The public parameters are made available to the adversary.

\item Given the public parameters $PP$, the adversary chooses a sequence of requests (of its choice) 
defined by $\{(\texttt{reqtype}_i,\texttt{metadata}_i)\}$ for $1\le i\le q$
($q$ is polynomial in the security parameter $\lambda$). The type of each of these requests, 
defined by $\texttt{reqtype}$, is an insertion (or revocation) of a certificate $c$ or a query $Q$ 
of any of the five types described above. The relevant information for each of these requests is 
stored in the corresponding $\texttt{metadata}$. If the request is for an insertion (or revocation) of a certificate, 
the challenger performs the necessary changes in the log structure and publishes the updated public parameters $PP$. 
If the request is a query $Q$, the challenger generates the proof corresponding to that particular $Q$ and sends it to the adversary.
\end{itemize}

Let $PP^*$ be the final public parameters at the end of the security game mentioned above. The adversary
generates a proof $\Pi$ of one of the five types. The adversary wins the game if the proof $\Pi$ is not
provided by the challenger in the request phase and Verify$(Q,\Pi,PP^*)=$ \texttt{accept}.

\paragraph{\bf Security Proof}\quad
Based on the security game described above, we show that an adversary in our scheme cannot win the 
game except with some probability negligible in the security parameter $\lambda$. 
To be more precise, the adversary cannot produce a valid proof of one of the following types 
unless it is provided by the challenger itself in the request phase.

\begin{itemize}
\item Proof of presence of a certificate (Type 1): A proof of this type consists of two sequences ($dataseq_{type1},hashseq_{type1}$)
in the searchTree. Since the hash function $h$ involved in the computation of the root digest of
the searchTree is \textit{collision-resistant} and the root digest is public,
the adversary fails to provide such a valid pair of sequences,
except with some probability negligible in $\lambda$.
    
\item Proof of absence of a certificate (Type 2): A proof of this type is a witness of
non-membership in the accumulator set $X$.
Since the bilinear-map accumulator used in our scheme is secure,
the adversary cannot forge a proof $\Pi$ showing that a certificate $c = cert(u, pk_u)$ is
absent in the set $X$ (where actually $c \in X$).
        
\item Proof of absence of a domain owner (Type 3): The proof is similar to the proof of Type 1.
A proof consists of two sequences ($dataseq_{type3},hashseq_{type3}$)
in the searchTree.
Since the hash function $h$ involved in the computation of the root digest of
the searchTree is \textit{collision-resistant} and the root digest is public, the adversary fails
to forge a proof of Type 3, except with some probability negligible in $\lambda$.
    
\item Proof of extension (Type 4): The proof of extension involves the computation of
the \textit{collision-resistant} hash function $h$ per layer of the chronTree and
checking whether the final hash value is equal to $digCT$. Due to the collision-resistance
property of $h$, the adversary cannot forge a proof for extension.
    
\item Proof of currency (Type 5): A proof of this type is a witness of membership
in the accumulator set $X$.
Since the bilinear-map accumulator used in our scheme is secure,
the adversary cannot forge a proof $\Pi$ showing that a certificate $c = cert(u, pk_u)$ is
present in the set $X$ (where actually $c \notin X$).
\end{itemize}
\end{proof}

\begin{lemma}
If a dishonest log maintainer maliciously provides a proof $\Pi$ of any type for a certificate $c = cert(u, pk^{\prime}_u)$,
then an auditor or the domain owner $u$ will be able to detect it.
\end{lemma}

\begin{proof}

As the dishonest log maintainer has the secret trapdoor value $s$, it can attempt to create a false proof of currency (Type 5)
or a false proof of absence (Type 2) for the certificate $c$ by computing the corresponding witness $w_c = A^{\frac{1}{c+s}}$
or $\hat{w}_c=(w_c, v_c)$ (for any such pair satisfying $w_c = A^{\frac{1}{c+s}}\cdot g^{\frac{v_c}{c+s}}$), respectively
(see the procedure Verify in Section~\ref{subsec:contribution}).
However, the log maintainer
sends an SCT to a certificate authority or a domain owner as a ``promise'' for this certificate to be included later.
As this certificate was never requested to be included in the log, the domain owner can easily detect log inconsistencies
while executing an efficient gossip protocol~\cite{NGR15,ChuatSPLM15}.

For the other proofs (Type 1, Type 3 and Type 4), the log maintainer cannot produce valid proofs with respect to the root
digest of the searchTree or the chronTree (included in the public parameters $PP$). This is due to the \textit{collision-resistance}
property of the hash function $h$. The proof is similar to that discussed in Lemma 3.

\end{proof}

\section{Performance Analysis}
\label{sec:perf_ana}
\subsection{Asymptotic Analysis}
\label{asym_ana}
Let $n$ be the total number of certificates in the log structure, $t$ be the total number of domain owners
and $m$ be the total number of active certificates (the size of the accumulator set $X$).
\paragraph{\bf Cost of Insertion and Revocation}\quad
In the chronTree, the cost of the insertion or revocation of a certificate is
$O(\log n)$ as the new leaf-node is to be inserted to the right of the chronTree
and the new root digest $digCT$ is to be computed.
In the searchTree, the cost of the insertion or revocation of a certificate is
$O(\log t)$ since searching for the node corresponding to the particular domain owner
(and computing the updated root digest $digST$) takes $O(\log t)$ time.
In the accTree (maintained for the accumulator $X$), an insertion or revocation takes $O(m)$ time 
as the new membership witness is to be updated (and stored) for each element of the set $X$.
Although the cost of revocation is $O(m)$ for the bilinear-map accumulator, we get
constant size proofs related to revocation transparency
with constant verification cost (discussed in the following sections).
However, this trade-off is justified as revocation is done by the powerful log maintainer (or certificate prover),
whereas the proof is verified by a lightweight auditor (or a monitor run by a domain owner).

\paragraph{\bf Parameters of a Proof}\quad
For each type of proof, we consider the following parameters: the size of the proof,
the computation cost for the proof and the verification time for the proof.
We compare our construction with the existing schemes for certificate transparency
based on these parameters.
The comparison is summarized in Table~\ref{tab:compar} (see Section~\ref{sec:intro}).

\begin{itemize}
 \item Proof of presence of a certificate (Type 1): The proof consists of the sequences $dataseq_{type1}$ and $hashseq_{type1}$
 in the searchTree. The size of each of these sequences, the time taken to generate them and the time
 taken to verify them with respect to the root digest $digST$ --- all are $O(\log t)$.

 \item Proof of absence of a certificate (Type 2): The size of a non-membership (in $X$) witness is $O(1)$.
 The computation time of a proof is $O(m)$ as the calculation of the non-membership witness
 requires $O(m)$ multiplications in $\Z_p^*$.
 The verification time for a proof is $O(1)$
 as it involves only two pairing operations.

 \item Proof of absence of a domain owner (Type 3): The proof is similar to the proof of Type 1.
 Here, the proof consists of the sequences $dataseq_{type3}$ and $hashseq_{type3}$
 in the searchTree.
 Thus, the size of a proof, the time taken to generate it and the time
 taken to verify it with respect to the root digest $digST$ --- all are $O(\log t)$.

 \item Proof of extension (Type 4): The proof consists of the sequence $(h_1, h_2, \ldots)$
 in the chronTree. The size of the sequence, the time taken to generate it and the time
 taken to verify it with respect to the root digest $digCT$ --- all are $O(\log n)$.

 \item Proof of currency (Type 5): The size of a membership (in $X$) witness is $O(1)$.
 The search for a node containing the certificate $c$ in the accTree takes $O(\log m)$ time,
 and it takes $O(1)$ time to retrieve the precomputed witness $w_c$ stored at that node.
 Thus, the computation time of a proof is $O(\log m)$.
 The verification time for a proof is $O(1)$ as it involves only two pairing operations.

\end{itemize}

\subsection{Performance Evaluation}
\label{subsec:perf_eval}
\subsubsection{Bilinear Setting and Hash Function}
\label{subsubsec:settings}\quad
In general setting, we take the bilinear pairing function $e:G_1\times G_2\rightarrow G_T$
with parameters $(p,g_1,g_2,G_1,G_2,G_T)$, where $|G_1|=|G_2|=|G_T|=p=\Theta(2^{2\lambda})$
and $g_1,g_2$ are generators of the groups $G_1$ and $G_2$, respectively.
We take $\lambda=128$.
Practical constructions of pairings are done on elliptic (or hyperelliptic) curves
defined over a finite field.
We write $E(\F_q)$ to denote the set of points on an elliptic curve $E$ defined over
the finite field $\F_q$. Then $G_1$ is taken as a subgroup of $E(\F_q)$, $G_2$ is taken
as a subgroup of $E(\F_{q^{k'}})$ and $G_T$ is taken as a subgroup of $\F_{q^{k'}}^*$,
where $k'$ is the embedding degree~\cite{KM_CC,Galbraith_DAM,Smart_DAM}.
For the value of the security parameter $\lambda$ up to 128, Barreto-Naehrig (BN)
curves~\cite{Barreto_SAC,Taxonmy_JOC} are suitable for our scheme.
In this setting, each of the elements of $\Z_p^*$ and $G_1$ is of size 256 bits.
We use SHA-256 as the \textit{collision-resistant} hash function $h$ used to compute
the root digests corresponding to the chronTree and the searchTree.

\subsubsection{Size of a Proof}
\label{subsubsec:proof_size}\quad
We calculate the size (in bits) of each type of proof as follows.

\begin{itemize}
 \item Proof of presence of a certificate (Type 1): %The proof consists of the sequences $(dataseq_{type1}$, $hashseq_{type1})$ in the searchTree.
 If SHA-256 is used as the hash function $h$, then the size of $hashseq_{type1}$ is at most $256\log t$ bits.
 Let the size of the data item stored in each node of the searchTree be denoted by $|data|$.
 Then, $|data|\approx N\cdot pk_{size}$, where $N$ is the maximum size of the list of public keys stored corresponding to
 a domain owner and $pk_{size}$ is the size of a public key. Thus, the size of a proof (in bits) is
 %\begin{equation*}
  $\log t(256+N\cdot pk_{size})$,
 %\end{equation*}
 where $t$ is the number of domain owners (or nodes) present in the searchTree.

 \item Proof of absence of a certificate (Type 2): The size of a proof is 512 bits as the size of
 a non-membership (in $X$) witness $\hat{w_c}=(w_c\in G_1,v_c\in\Z_p^*)$ is 512 bits.

 \item Proof of absence of a domain owner (Type 3): The proof is similar to the proof of Type 1.
 Thus, the size of a proof is $\log t(256+\log t+N\cdot pk_{size})$ bits,
 where $t$ is the number of domain owners (or nodes) present in the searchTree.

 \item Proof of extension (Type 4): The proof consists of a sequence of hash values
 in the chronTree. Therefore, for SHA-256 used as the hash function $h$, the size of
 a proof is at most $256\log n$ bits, where $n$ is the number of certificates present in the chronTree.

 \item Proof of currency (Type 5): The size of a proof is 256 bits as
 the size of a membership (in $X$) witness $w_c\in G_1$ is 256 bits.

\end{itemize}

\subsubsection{Cost of Verification of a Proof}
\label{subsubsec:vercost}\quad
For the timing analysis of the pairing operations, we use PandA, a recent
software framework for \textit{Pairings and Arithmetic} developed by
Chuengsatiansup \emph{et al.}~\cite{Panda_PBC}.
The cycle-counts of a pairing operation (Ate pairing) for a 128-bit secure, Type-3 pairing framework
involving a pairing-friendly BN curve is 3832644 (taken as the median of 10000 measurements
on a 2.5 GHz Intel Core i5-3210M processor~\cite{Panda_PBC}).
These many cycles take approximately 1.53 milliseconds on this processor.
On the other hand, we estimate the time taken, on a 2.5 GHz Intel Core i5-3210M processor,
to compute SHA-256 from \textit{eBASH}, a benchmarking project for hash functions~\cite{eBASH}.

\begin{itemize}
 \item Proof of presence of a certificate (Type 1): The proof consists of the sequences $dataseq_{type1}$ and $hashseq_{type1}$
 in the searchTree. The number of hashes performed for the verification is at most $\log t$.
 Each hash is performed on a message of the form $(d,h_1,h_2)$, where $d$ is the data item associated
 with a node and $h_1$ (or $h_2$) is the hash value of the left (or right) child of the node.
 Let the size of the data item stored in each node of the searchTree be denoted by $|data|$.
 Then, $|data|\approx N\cdot pk_{size}$,
 where $N$ is the maximum size of the list
 of public keys stored corresponding to a domain owner and $pk_{size}$ is the size of a public key.
 For SHA-256 used as the hash function $h$, the size of $h_1$ (or $h_2$) is $256$ bits.
 Therefore, the size of the message input to the hash function $h$ is given
 by $hash_{in}=(512+|data|)$ bits.
 For long messages, computing a single hash value takes 12.71 cycles per byte of the message
 on a 2.5 GHz Intel Core i5-3210M processor~\cite{eBASH}.
 Consequently, computing each hash value requires
 $(hash_{in}\cdot (1.59))$ cycles (approximately) that, in turn, takes around
 $(hash_{in}\cdot (0.64)\cdot 10^{-6})$ milliseconds on this processor.
 Thus, the total time required for the verification of a proof is around
 $(\log t\cdot hash_{in}\cdot (0.64)\cdot 10^{-6})$ milliseconds.

 \item Proof of absence of a certificate (Type 2): 
 The cost of verification of a proof is 3.06 milliseconds
 as the verification requires two pairing operations.

 \item Proof of absence of a domain owner (Type 3): The proof is similar to that of Type 1.
 Thus, the time required for the verification of a proof is around
 $(\log t\cdot hash_{in}\cdot (0.64)\cdot 10^{-6})$ milliseconds.

 \item Proof of extension (Type 4): The proof consists of a sequence of hash values
 in the chronTree. If SHA-256 is used as the hash function $h$, then the size of
 the input to $h$ is $512$ bits (the hash values of two children). The total number
 of hashes to be performed is at most $\log n$, where $n$ is the number of certificates present in the chronTree.
 For messages of length 64 bytes, computing a single hash value takes 29.94 cycles per byte of the message
 on a 2.5 GHz Intel Core i5-3210M processor~\cite{eBASH}.
 Thus, the total time required for the verification of a proof is around
 $(\log n\cdot 64\cdot (11.98)\cdot 10^{-6})$ milliseconds.

 \item Proof of currency (Type 5): 
 The cost of verification of a proof is 3.06 milliseconds
 as the verification requires two pairing operations.
 
\end{itemize}

\section{Conclusion}
\label{sec:conclusion}
We have developed a scheme which is an extended version of the existing certificate transparency schemes.
Some of the proofs in our scheme enjoy constant proof-size and constant
verification cost.
Apart from handling the existing proofs efficiently, we have introduced a new proof confirming the absence of
a particular certificate in the log structure. We have also analyzed the security of our scheme.
Finally, we have provided a thorough performance evaluation of our scheme.

\section{Acknowledgments}
This project has been made possible in part by a gift from the NetApp University Research Fund,
a corporate advised fund of Silicon Valley Community Foundation.

\bibliographystyle{splncs03}
\bibliography{draft}

\end{document}